\documentclass[conference,a4paper,final]{IEEEtran}

\usepackage[dvips]{graphicx}
\usepackage[cmex10]{amsmath}
\usepackage{amssymb}
\usepackage{amsthm}
\usepackage{array}
\usepackage{mdwmath}
\usepackage{mdwtab}
\usepackage[tight,footnotesize]{subfigure}
\usepackage{fixltx2e}

\usepackage{psfragx}

\theoremstyle{plain}
\newtheorem{theorem}{Theorem}
\newtheorem{lemma}{Lemma}
\newtheorem{corollary}{Corollary}

\renewcommand{\vec}[1]{\mathbf{#1}}  
\newcommand{\mat}[1]{\mathbf{#1}}

\begin{document}

\sloppy

\title{Feasibility Conditions of Interference Alignment\\via Two Orthogonal Subcarriers}

\author{
  \IEEEauthorblockN{Stefan Dierks and Gerhard Kramer}
  \IEEEauthorblockA{Institute for Communications Engineering\\
    Technische Universit\"{a}t M\"{u}nchen,\\
    Munich, Germany\\
    Email:  \{stefan.dierks, gerhard.kramer\}@tum.de} 
  \and
  \IEEEauthorblockN{Wolfgang Zirwas}
  \IEEEauthorblockA{Nokia Siemens Networks\\
    Munich, Germany\\
    Email: wolfgang.zirwas@nsn.com}
}



\maketitle

\begin{abstract}
Conditions are derived on line-of-sight channels to ensure the feasibility of interference alignment. The conditions involve choosing only the spacing between two subcarriers of an orthogonal frequency division multiplexing (OFDM) scheme.
The maximal degrees-of-freedom are achieved and even an upper bound on the sum-rate of interference alignment is approached arbitrarily closely.

\end{abstract}

\section{Introduction}

Interference alignment (IA) is a promising method because it achieves higher throughput in interference limited scenarios than conventional methods such as time- or frequency-division multiplexing or treating interference as noise \cite{Jafar2010}.
The main idea of IA is to use precoding at the transmitters to align interference at each receiver in one subspace. The orthogonal subspace is used for interference-free communication.

One commonly measures performance by the sum of the rates that the users can transmit reliably. The degrees-of-freedom (DoF) are defined as
\begin{IEEEeqnarray}{rCl}
d &=& \lim_{\text{SNR}\rightarrow\infty}\frac{C_\text{sum}(\text{SNR})}{\log(\text{SNR})},
\end{IEEEeqnarray}
where $C_\text{sum}(\text{SNR})$ is the sum-rate capacity at the signal to noise ratio $\text{SNR}$.
The DoF represent the number of non-interfering data streams that can be simultaneously transmitted over the network.

For single antennas IA achieves the maximal DoF asymptotically with an \emph{infinite} number of subcarriers or time-slots \cite{Cadambe2008}. We derive conditions for which \emph{two} subcarriers make IA feasible for general channels in Section \ref{sec:gen_ch} and for line-of-sight channels in Section \ref{sec:single-tap}. For line-of-sight (i.e. single-tap) channels these conditions are fulfilled by choosing the subcarrier spacing carefully, while in prior art the subcarriers are assumed to be fixed when IA is applied. Hence for line-of-sight channels we achieve the maximal DoF and even achieve an upper bound on the sum-rate of IA arbitrarily closely.



\section{System Model}
\label{sec:sys_mod}

Consider an interference channel with $K$ user pairs, where each transmitter sends either one or two streams to its receiver. Each node is equipped with a single antenna and uses the same two orthogonal subcarriers. The received signal in the frequency domain at receiver $i$ is
\begin{IEEEeqnarray}{rCl}
\label{eq:sys_mod}
\vec{y}_{i} &=& \mat{U}_{i}^{\dagger} \mat{H}_{i,i} \mat{V}_{i} \vec{s}_{i} + \left[ \sum_{k=1, k\neq i}^{K} \mat{U}_{i}^{\dagger} \mat{H}_{i,k} \mat{V}_{k} \vec{s}_{k} \right] + \mat{U}_{i}^{\dagger}\vec{z}_{i}
\end{IEEEeqnarray}
where $\vec{s}_{k}$ is the vector of symbols at transmitter $k$ with length $D\in\{1,2\}$, $\mat{V}_{k}$ is a $2\times D$ precoding matrix, $\mat{H}_{i,k}$ is a $2 \times 2$ channel matrix in the frequency domain between transmitter $k$ and receiver $i$, $\mat{U}_{i}$ is a $D\times 2$ receive filter matrix. $\mat{U}^\dagger$  is the complex conjugate transpose of matrix $\mat{U}$, while $u^\ast$ is the complex conjugate of scalar $u$.
The precoding and the receive filter matrices are chosen to satisfy $\Vert\mat{V}_{k}\Vert_{F} \leq 1$ and $\Vert\mat{U}_{i}\Vert_{F} \leq 1$, where $\Vert\; \Vert_{F}$ denotes the Frobenius norm.
$\vec{z}_{i}$ is a proper complex AWGN vector of length $2$ and variance $\sigma^2_{i}$. The first term on the right hand side of \eqref{eq:sys_mod} carries the data of receiver $i$, while the sum represents the interference, and the last term is filtered noise. Channels connecting the transmitter and receiver of the same user pair are called \emph{direct} channels; the other channels (i.e. $\mat{H}_{i,k}$ $i\neq k$) are called \emph{cross} channels.

For orthogonal subcarriers the channel matrices $\mat{H}$ are diagonal. The diagonal entries are denoted by $h_{i,k}^{(l)} \in \mathbb{C}$, where $l$ indicates the subcarrier index. We write
\begin{IEEEeqnarray}{rcccl}
\mat{H}_{i,k} &=& \begin{bmatrix}
                h_{i,k}^{(1)} \!& 0\\
		0 \!& h_{i,k}^{(2)}
                \end{bmatrix}
		&=& \begin{bmatrix}
                \left|h_{i,k}^{(1)}\right|e^{-j\angle h_{i,k}^{(1)}} \!& 0\\
		0 \!& \left|h_{i,k}^{(2)}\right|e^{-j\angle h_{i,k}^{(2)}}
                \end{bmatrix}\!,\quad
\end{IEEEeqnarray}
where $\left|x\right|$ denotes the amplitude of $x$ and $\angle x$ denotes the phase of $x$ in radians.
For line-of-sight channels the amplitudes are equal for all subcarriers, while the phase rotations depend on the delay $\tau_{i,k}$ and the subcarrier frequencies $f^{(1)}$ and $f^{(2)}$:
\begin{IEEEeqnarray}{rcl}
\mat{H}^\text{LoS}_{i,k} &=& \left|h_{i,k}\right|\begin{bmatrix}
                e^{-j 2\pi f^{(1)} \tau_{i,k}} \!& 0\\
		0 \!& e^{-j 2\pi f^{(2)} \tau_{i,k}}
                \end{bmatrix}\!.\quad
\end{IEEEeqnarray}
The amplitudes are bounded as $0<\left|h_{i,k}^{(l)}\right|$ to avoid degenerate channel conditions.
We assume perfect channel knowledge of all channel parameters at all nodes.

%

\section{Feasibility of Interference Alignment via Two Subcarriers}
\label{sec:gen_ch}
For single antenna nodes, the DoF are upper-bounded by $1/2$ per user pair \cite{Cadambe2008}. The precoder and receive filters reduce to vectors $\vec{v}_{k}$ and $\vec{u}_{i}$ and interference is aligned if \cite{Gomadam2011}
\begin{IEEEeqnarray}{rCll}
\label{eq:ia_crossCh}
\vec{u}_{i}^{\dagger}\mat{H}_{i,k}\vec{v}_{k} &=& \mat{0}, &\quad\forall i \neq k\\
\label{eq:ia_directCh}
\left|\vec{u}_{i}^{\dagger}\mat{H}_{i,i}\vec{v}_{i}\right| &>& 0, &\quad\forall i.
\end{IEEEeqnarray}
The equations \eqref{eq:ia_crossCh} mean that the interference lies in the null-space of the receive filter, while the equations \eqref{eq:ia_directCh} ensure that the effective channel $\bar{h}_{i}=\vec{u}_{i}^{\dagger}\mat{H}_{i,i}\vec{v}_{i}$ (which is interference-free if the first set of equations is fulfilled) has unit rank.
The question of IA feasibility asks if there is a solution for $\vec{u}_{i}$ $\forall i$ and $\vec{v}_{k}$ $\forall k$ such that \eqref{eq:ia_crossCh} and \eqref{eq:ia_directCh} are fulfilled.

Suppose that all channel coefficients are chosen independently with a continuous distribution. The conditions \eqref{eq:ia_directCh} are fulfilled with probability $1$ if the precoder and receive filters satisfy \eqref{eq:ia_crossCh}. Hence we need to examine the feasibility of \eqref{eq:ia_crossCh} to show that the maximal DoF are achievable.

The question of feasibility is tackled, e.g., in \cite{Cadambe2008}\nocite{Gomadam2011}\nocite{Cadambe2007}-\cite{Yetis2009}.
In \cite{Cadambe2008} it is shown that the maximal DoF are asymptotically achievable with IA for time-varying channels by increasing the number of symbol extensions (i.e. the number of subcarriers or time slots).
We show that for interesting channel conditions IA is feasible with \emph{two} subcarriers. For this we use Lemma \ref{lem:cross_ch} to write the IA equation set \eqref{eq:ia_crossCh} as the sum of logarithms of the channel, precoder, and receive filter coefficient fractions. With Lemma \ref{lem:cross_ch} we prove Theorem \ref{th:3user} which states feasibility conditions on the channel coefficients for the 3 user pairs case. In the following subsection we examine $K$ user pairs.

\begin{lemma}
\label{lem:cross_ch}
 For single antenna nodes and two orthogonal subcarriers the IA conditions \eqref{eq:ia_crossCh} are
 \begin{IEEEeqnarray}{rCl}
  \ln\!\left(\! \frac{u_{i}^{(2)\ast}}{u_{i}^{(1)\ast}} \!\right)
  \!+\! \ln\!\left(\! \frac{v_{k}^{(2)}}{v_{k}^{(1)}} \!\right)
  &=&j\pi\left(1\!+\!2n_{i,k}\right)
  \!-\! \ln\!\left(\! \frac{h_{i,k}^{(2)}}{h_{i,k}^{(1)}} \!\right)
  \label{eq:IAasSum}
 \end{IEEEeqnarray}
 for all $i\neq k$, where $n_{i,k}\in\mathbb{Z}$ can be any integer.
\end{lemma}

\begin{proof}
We write \eqref{eq:ia_crossCh} as the equation set
\begin{IEEEeqnarray}{rCl}
u_{i}^{(1)\ast} h_{i,k}^{(1)} v_{k}^{(1)} 
+ u_{i}^{(2)\ast} h_{i,k}^{(2)} v_{k}^{(2)} 
\label{eq:IA2sc}
&=& 0, \quad \forall i \neq k.
\end{IEEEeqnarray}
There exist trivial solutions of \eqref{eq:IA2sc}:
\begin{itemize}
 \item $\vec{u}_{i} = \vec{0}$ or $\vec{v}_{i}= \vec{0}$, which both violate \eqref{eq:ia_directCh};
 \item $u_{i}^{(1)} \!=\! v_{k}^{(2)} \!=\! 0$ or $v_{i}^{(1)} \!=\! u_{k}^{(2)} \!=\! 0$, which, when examining the equation set $\forall i \neq k$, lead to the invalid solutions $\vec{u}_{i} = \vec{0}$  $\forall i$ or $\vec{v}_{k} = \vec{0}$ $\forall k$.
\end{itemize}
Other trivial solutions with $u_{i}^{(l)} \!=\! 0$ or $v_{k}^{(l)} \!=\! 0$ do not exist, since we have $h_{i,k}^{(1)}\neq0$ and $h_{i,k}^{(2)}\neq0$ (recall that $\left| h_{i,k}^{(l)}\right|>0$). Hence all $u_{i}$ and $v_{k}$ are non-zero for nontrivial solutions.

Manipulating \eqref{eq:IA2sc} we obtain
\begin{IEEEeqnarray}{rCl}
\frac{u_{i}^{(2)\ast} h_{i,k}^{(2)} v_{k}^{(2)}}{u_{i}^{(1)\ast} h_{i,k}^{(1)} v_{k}^{(1)}}&=&-1
\end{IEEEeqnarray}
and therefore
\begin{IEEEeqnarray}{rCl}
\ln\left( \frac{u_{i}^{(2)\ast}h_{i,k}^{(2)}v_{k}^{(2)}}{u_{i}^{(1)\ast}h_{i,k}^{(1)}v_{k}^{(1)}} \right)
&=&j\pi\left(1+2n_{i,k}\right)
\end{IEEEeqnarray}
where $n_{i,k}\in\mathbb{Z}$.
\end{proof}

Note that $h_{i,k}^{(1)}=0$ or $h_{i,k}^{(2)}=0$ have zero probability for continuous distributions.

\subsection{3 User Pairs}

\begin{theorem}
\label{th:3user}
 Three DoF over two subcarriers are feasible for three user pairs with single antennas if the following condition holds
  \begin{IEEEeqnarray}{rCl}
  \frac{h_{1,2}^{(2)}}{h_{1,2}^{(1)}}\frac{h_{1,3}^{(1)}}{h_{1,3}^{(2)}}
  \frac{h_{2,3}^{(2)}}{h_{2,3}^{(1)}}\frac{h_{2,1}^{(1)}}{h_{2,1}^{(2)}}
  \frac{h_{3,1}^{(2)}}{h_{3,1}^{(1)}}\frac{h_{3,2}^{(1)}}{h_{3,2}^{(2)}}&=&1.
  \label{eq:th_3user}
   \end{IEEEeqnarray}
\end{theorem}

\begin{proof}
 \footnote{The proof can also be obtained by examining the subspaces spanned by the channel matrix and the precoding vector as is done in Section IV-D of \cite{Cadambe2008}. For interference to align one must have $\text{span}\!\left(\mat{H}_{1,2}\vec{v}_{2}\right) = \text{span}\!\left(\mat{H}_{1,3}\vec{v}_{3}\right)$ and $\text{span}\!\left(\mat{H}_{2,3}\vec{v}_{3}\right) = \text{span}\!\left(\mat{H}_{2,1}\vec{v}_{1}\right)$ and $\text{span}\!\left(\mat{H}_{3,1}\vec{v}_{1}\right) = \text{span}\!\left(\mat{H}_{3,2}\vec{v}_{2}\right)$. From this one obtains $\text{span}\!\left(\vec{v}_{1}\right) = \text{span}\!\left(\mat{T}\vec{v}_{1}\right)$, where $\mat{T}= \left(\mat{H}_{1,3}\right)^{-1}\mat{H}_{2,3} \left(\mat{H}_{2,1}\right)^{-1}\mat{H}_{1,2} \left(\mat{H}_{3,2}\right)^{-1}\mat{H}_{3,1}$. Due to the diagonal structure of the channel matrices, $\mat{T}$ is also diagonal. Unless $\mat{T}$ is a (scaled) identity matrix the precoder $\vec{v}_{1}$ must be an eigenvector of all channel matrices, leading to interference not being aligned. Setting $\mat{T}$ as a scaled identity matrix leads to \eqref{eq:th_3user}.}
 For three users there are six cross-channels. According to Lemma \ref{lem:cross_ch} six equations of type \eqref{eq:IAasSum} must be satisfied.
 We write these equations in the form $\mat{A}\vec{x} = \vec{b}$ as follows:
 \begin{IEEEeqnarray}{rcl}
  \underbrace{
		  \begin{bmatrix}
		  \!1\!\! & \!0\!\! & \!0\!\! & \!0\!\! & \!1\!\! & \!0\!\\
		  \!1\!\! & \!0\!\! & \!0\!\! & \!0\!\! & \!0\!\! & \!1\!\\
		  \!0\!\! & \!1\!\! & \!0\!\! & \!0\!\! & \!0\!\! & \!1\!\\
		  \!0\!\! & \!1\!\! & \!0\!\! & \!1\!\! & \!0\!\! & \!0\!\\
		  \!0\!\! & \!0\!\! & \!1\!\! & \!1\!\! & \!0\!\! & \!0\!\\
		  \!0\!\! & \!0\!\! & \!1\!\! & \!0\!\! & \!1\!\! & \!0\!\\
		  \end{bmatrix}\!}_{\text{Rank}(\mat{A}) = 5}\!
		  \begin{bmatrix}
		  \!\ln\!\left(\! \frac{u^{(2)\ast}_{1}}{u^{(1)\ast}_{1}}\!\right)\!\\
		  \!\ln\!\left(\! \frac{u^{(2)\ast}_{2}}{u^{(1)\ast}_{2}}\!\right)\!\\
		  \!\ln\!\left(\! \frac{u^{(2)\ast}_{3}}{u^{(1)\ast}_{3}}\!\right)\!\\
		  \!\ln\!\left(\! \frac{v^{(2)}_{1}}{v^{(1)}_{1}}\!\right)\!\\
		  \!\ln\!\left(\! \frac{v^{(2)}_{2}}{v^{(1)}_{2}}\!\right)\!\\
		  \!\ln\!\left(\! \frac{v^{(2)}_{3}}{v^{(1)}_{3}}\!\right)\!
		  \end{bmatrix}
  &\!=\!&
		  \begin{bmatrix}
		  j\pi\!\left(\!1\!+\!2n_{1,2}\!\right)\!-\!\ln\!\left(\! \frac{h_{1,2}^{(2)}}{h_{1,2}^{(1)}}\!\right)\!\\
		  j\pi\!\left(\!1\!+\!2n_{1,3}\!\right)\!-\!\ln\!\left(\! \frac{h_{1,3}^{(2)}}{h_{1,3}^{(1)}}\!\right)\!\\
		  j\pi\!\left(\!1\!+\!2n_{2,3}\!\right)\!-\!\ln\!\left(\! \frac{h_{2,3}^{(2)}}{h_{2,3}^{(1)}}\!\right)\!\\
		  j\pi\!\left(\!1\!+\!2n_{2,1}\!\right)\!-\!\ln\!\left(\! \frac{h_{2,1}^{(2)}}{h_{2,1}^{(1)}}\!\right)\!\\
		  j\pi\!\left(\!1\!+\!2n_{3,1}\!\right)\!-\!\ln\!\left(\! \frac{h_{3,1}^{(2)}}{h_{3,1}^{(1)}}\!\right)\!\\
		  j\pi\!\left(\!1\!+\!2n_{3,2}\!\right)\!-\!\ln\!\left(\! \frac{h_{3,2}^{(2)}}{h_{3,2}^{(1)}}\!\right)\!
		  \end{bmatrix}\!.\qquad
	\label{eq:IA_linSys}
 \end{IEEEeqnarray}
 Since the rank of $\mat{A}$ is $5$, which is less than the number of equations, a solution exists if and only if the rank of the augmented matrix $\left(\mat{A}|\vec{b}\right)$ is equal to the rank of $\mat{A}$ (or $\vec{b}$ is in the column space or image of $\mat{A}$). This condition is fulfilled for
  \begin{IEEEeqnarray}{RcCl}
  \notag
  & \ln\!\left(\! \frac{h_{1,2}^{(2)}}{h_{1,2}^{(1)}}\!\right)
  - \ln\!\left(\! \frac{h_{1,3}^{(2)}}{h_{1,3}^{(1)}}\!\right)
  + \ln\!\left(\! \frac{h_{2,3}^{(2)}}{h_{2,3}^{(1)}}\!\right)\\
  \label{eq:con_3user}
  -& \ln\!\left(\! \frac{h_{2,1}^{(2)}}{h_{2,1}^{(1)}}\!\right)
  + \ln\!\left(\! \frac{h_{3,1}^{(2)}}{h_{3,1}^{(1)}}\!\right)
  - \ln\!\left(\! \frac{h_{3,2}^{(2)}}{h_{3,2}^{(1)}}\!\right) &=& j2\pi n,
   \end{IEEEeqnarray}
 where $n=n_{1,2}-n_{1,3}+n_{2,3}-n_{2,1}+n_{3,1}-n_{3,2}\in\mathbb{Z}$.
\end{proof}

Theorem \ref{th:3user} can be expressed as two equations: One for the subcarrier amplitudes
\begin{IEEEeqnarray}{rCl}
\label{eq:IA3user_amp}
\frac{\left|h_{1,2}^{(2)}\right|}{\left|h_{1,2}^{(1)}\right|}\frac{\left|h_{1,3}^{(1)}\right|}{\left|h_{1,3}^{(2)}\right|}
\frac{\left|h_{2,3}^{(2)}\right|}{\left|h_{2,3}^{(1)}\right|}\frac{\left|h_{2,1}^{(1)}\right|}{\left|h_{2,1}^{(2)}\right|}
\frac{\left|h_{3,1}^{(2)}\right|}{\left|h_{3,1}^{(1)}\right|}\frac{\left|h_{3,2}^{(1)}\right|}{\left|h_{3,2}^{(2)}\right|} &=& 1
\end{IEEEeqnarray}
and one for the subcarrier phase rotations
\begin{IEEEeqnarray}{rCl}
\notag
&-& \angle h_{1,2}^{(2)} + \angle h_{1,2}^{(1)}
+ \angle h_{1,3}^{(2)} - \angle h_{1,3}^{(1)}
- \angle h_{2,3}^{(2)} + \angle h_{2,3}^{(1)}\\
\notag
&+& \angle h_{2,1}^{(2)} - \angle h_{2,1}^{(1)}
- \angle h_{3,1}^{(2)} + \angle h_{3,1}^{(1)}
+ \angle h_{3,2}^{(2)} - \angle h_{3,2}^{(1)}\\
\label{eq:IA3user_pha}
&=& 2\pi n.
\end{IEEEeqnarray}

\subsection{$K$ User Pairs}
\label{sec:k-user}

For $K$ user pairs there are $K(K-1)$ cross-channels and we hence have $K(K-1)$ equations of type \eqref{eq:IAasSum}. We collect them into an equation system $\mat{A}\vec{x} = \vec{b}$, where $\mat{A}$ is of dimension $K(K-1) \times 2K$, but has rank $2K-1$.

The augmented matrix $\left(\mat{A}|\vec{b}\right)$ again must have the same rank as $\mat{A}$ for a solution to exist. Transforming $\mat{A}$ to row-echelon form by using Gaussian elimination results in a new matrix $\mat{A}^\prime$ where the last
 \begin{IEEEeqnarray}{rCl}
 K(K-1) - (2K-1) &=& K^2-3K+1
 \end{IEEEeqnarray}
rows are zero. We apply the same transformations to $\vec{b}$ to obtain $\vec{b}^\prime$. The last $K^2-3K+1$ entries of $\vec{b}^\prime$ must be zero, and are of the form
 \begin{IEEEeqnarray}{c}
 \sum_{\forall i,k} \alpha_{i.k}^{[w]}\left(j\pi\left(1+2n_{i,k}\right) - \ln\!\left( \frac{h_{i,k}^{(2)}}{h_{i,k}^{(1)}} \right)\right)
  \end{IEEEeqnarray}
where $\alpha_{i.k}^{[w]}\in\left\{-1,0,1\right\}$ are the weights of the $w$-th row.
Hence we obtain $K^2-3K+1$ equations of type similar to \eqref{eq:con_3user} which must be fulfilled for feasibility of IA.

\section{Special case: 3 User Pairs and Line-of-Sight Channels}
\label{sec:single-tap}

We examine IA for the special case of line-of-sight channels and $K=3$. We show that the feasibility condition of the channel can be fulfilled by choosing the subcarrier spacing carefully. We also derive the amplitudes of the effective channels and show that for increasing bandwidth an upper bound on the sum-rate of the presented scheme can be reached arbitrary closely.

\begin{corollary}
\label{cor:1tap}
 For line-of-sight channels the condition of Theorem \ref{th:3user} simplifies to
   \begin{IEEEeqnarray}{rCL}
  \label{eq:IA3user_los}
  \left(\! f^{(2)} \!-\! f^{(1)} \!\right) \left( \tau_{1,3} \!-\! \tau_{1,2} \!+\! \tau_{2,1} \!-\! \tau_{2,3} \!+\! \tau_{3,2} \!-\! \tau_{3,1} \right) &=& n\quad
 \end{IEEEeqnarray}
 where $n\in \mathbb{Z}\setminus \{0\}$.
\end{corollary}

\begin{proof}
 For single tap channels the subcarrier amplitudes satisfy $\left|h_{i,k}^{(1)}\right| \!=\! \left|h_{i,k}^{(2)}\right|$ and hence only the phase rotation difference remains. Inserting $\angle h_{i,k}^{(l)} = 2\pi f^{(l)} \tau_{i,k}$ gives
   \begin{IEEEeqnarray}{rcCl}
  \notag
  -& 2\pi \left(f^{(2)} - f^{(1)}\right) \tau_{1,2}
  + 2\pi \left(f^{(2)} - f^{(1)}\right) \tau_{1,3}\\
  -& 2\pi \left(f^{(2)} - f^{(1)}\right) \tau_{2,3}
  + 2\pi \left(f^{(2)} - f^{(1)}\right) \tau_{2,1}\\
  \notag
  -& 2\pi \left(f^{(2)} - f^{(1)}\right) \tau_{3,1}
  + 2\pi \left(f^{(2)} - f^{(1)}\right) \tau_{3,2}  &=& 2\pi n.
    \end{IEEEeqnarray}
 After some manipulations one obtains \eqref{eq:IA3user_los}.
 Choosing $n=0$ violates the assumption of orthogonal sub-carriers,
 since this means $f^{(2)} = f^{(1)}$.
\end{proof}

According to \eqref{eq:IA3user_los} line-of-sight channels may create conditions where IA is feasible by choosing the sub-carrier spacing $\Delta f = f^{(2)} - f^{(1)}$ carefully. This means that the precoding and receive filter vectors can be chosen such that \eqref{eq:ia_crossCh} holds. The required spacing depends only on the delays of the cross channels and the non-zero integer $n$ which can be chosen freely. Hence we can identify a minimal sub-carrier spacing
 \begin{IEEEeqnarray}{rCl}
 \Delta f_\text{min} &=& 1/\left( \tau_{1,3}-\tau_{1,2}+\tau_{2,1}-\tau_{2,3}+\tau_{3,2}-\tau_{3,1} \right)\quad
  \end{IEEEeqnarray}
for which IA is feasible. Any multiple of $\Delta f_\text{min}$, except $0$, creates feasibility again.

For the special case
 \begin{IEEEeqnarray}{rCl}
 \tau_{1,3}-\tau_{1,2}+\tau_{2,1}-\tau_{2,3}+\tau_{3,2}-\tau_{3,1}&=&0
  \end{IEEEeqnarray}
IA is directly feasible and the subcarrier spacing can be chosen arbitrarily. For continuously and independently distributed delays the probability of this event is zero and is not treated further.

Note that we are not limited to using two subcarriers. Since the feasibility depends solely on the spacing, subcarrier pair $f^{(1)}+f_\text{offset}$ and $f^{(2)}+f_\text{offset}$ is feasible if pair $f^{(1)}$ and $f^{(2)}$ is. Even different user pairs, which require different $\Delta f_\text{min}$, can be scheduled in one OFDM frame. It might not be possible to use all subcarriers with IA in which case the remaining subcarriers are used as usual.

\subsection{Effective Channel Amplitudes}
If \eqref{eq:IA3user_los} is fulfilled, the ratios of the precoding and receive filter coefficients are obtained from the system of linear equations \eqref{eq:IA_linSys}. Since $\mat{A}$ is rank-deficient there is one independent variable in $\vec{x}$, which we choose without loss of generality to be $\ln\!\left( u^{(2)\ast}_{1}/u^{(1)\ast}_{1} \right)$. The remaining variables are determined as
 \begin{IEEEeqnarray}{rCl}
\label{eq:x_v2}
 \ln\!\left(\! \frac{v^{(2)}_{2}}{v^{(1)}_{2}}\!\right)
 &=&
 j\pi\!\left(\!1\!+\!2n_{1,2}\!+\!2\Delta f \tau_{1,2}\!\right)\!-\!\ln\!\left(\! \frac{u^{(2)\ast}_{1}}{u^{(1)\ast}_{1}}\!\right)
\\
\label{eq:x_v3}
 \ln\!\left(\! \frac{v^{(2)}_{3}}{v^{(1)}_{3}}\!\right)
 &=&
 j\pi\!\left(\!1\!+\!2n_{1,3}\!+\!2\Delta f \tau_{1,3}\!\right)\!-\!\ln\!\left(\! \frac{u^{(2)\ast}_{1}}{u^{(1)\ast}_{1}}\!\right)
\\
\label{eq:x_u2}
 \ln\!\left(\! \frac{u^{(2)\ast}_{2}}{u^{(1)\ast}_{2}}\!\right)
 &=&
 j\pi\!\left(\!1\!+\!2n_{2,3}\!+\!2\Delta f \tau_{2,3}\!\right)\!-\!\ln\!\left(\! \frac{v^{(2)}_{3}}{u^{(1)}_{3}}\!\right)
\\
\label{eq:x_u3}
 \ln\!\left(\! \frac{u^{(2)\ast}_{3}}{u^{(1)\ast}_{3}}\!\right)
 &=&
 j\pi\!\left(\!1\!+\!2n_{3,2}\!+\!2\Delta f \tau_{3,2}\!\right)\!-\!\ln\!\left(\! \frac{v^{(2)}_{2}}{v^{(1)}_{2}}\!\right)
\\
\label{eq:x_v1}
 \ln\!\left(\! \frac{v^{(2)}_{1}}{v^{(1)}_{1}}\!\right)
 &=&
 j\pi\!\left(\!1\!+\!2n_{2,1}\!+\!2\Delta f \tau_{2,1}\!\right)\!-\!\ln\!\left(\! \frac{u^{(2)\ast}_{2}}{u^{(1)\ast}_{2}}\!\right).
\end{IEEEeqnarray}
From \eqref{eq:x_v2}-\eqref{eq:x_v1} one obtains, for $i\in\left\{1,2,3\right\}$,
\begin{IEEEeqnarray}{rCl}
\left|u_{i}^{(1)\ast}\right| \left|v_{i}^{(1)}\right| &=& \left|u_{i}^{(2)\ast}\right| \left|v_{i}^{(2)}\right|.
\end{IEEEeqnarray}
Together with $\Vert\vec{v}_{i}\Vert_{F} \leq 1$ and $\Vert\vec{u}_{i}\Vert_{F} \leq 1$ one obtains
\begin{IEEEeqnarray}{rCcCl}
\left|u_{i}^{(1)\ast} v_{i}^{(1)}\right| &\leq& \left|v_{i}^{(1)}\right|\sqrt{1- \left|v_{i}^{(1)}\right|^2} &\leq& 1/2.
\end{IEEEeqnarray}
For all else held fixed the $i$-th amplitude is largest if 
\begin{IEEEeqnarray}{rCcCcCcCl}
\label{eq:amp_larg}
\left|v_{i}^{(1)}\right| &=& \left|v_{i}^{(2)}\right| &=& \left|u_{i}^{(1)\ast}\right| &=& \left|u_{i}^{(2)\ast}\right| &=& 1/\sqrt{2}
\end{IEEEeqnarray}
which we use when obtaining the amplitudes.

The amplitude of the first direct channel is
\begin{IEEEeqnarray}{rcl}
\notag
\left|\bar{h}_{1}\right| &=& \left|\vec{u}_{1}^{\dagger}\mat{H}_{1,1}\vec{v}_{1}\right|
\\
\notag
&=& \left|h_{1,1}\right| \left| u_{1}^{(1)\ast} e^{-j2\pi f^{(1)} \tau_{1,1}} v_{1}^{(1)} + u_{1}^{(2)\ast} e^{-2\pi f^{(2)} \tau_{1,1}} v_{1}^{(2)} \right|
\\
\notag
&\stackrel{\text{(a)}}{=}& \frac{\left|h_{1,1}\right|}{2}
\left| 1 \!+\! e^{-j2\pi \Delta f \tau_{1,1} + \ln\!\left( u^{(2)\ast}_{1}\!/u^{(1)\ast}_{1}\right) + \ln\!\left( v^{(2)}_{1}\!/v^{(1)}_{1}\right) } \right|
\\
\notag
&\stackrel{\text{(b)}}{=}& \frac{\left|h_{1,1}\right|}{2}
\left| 1 \!-\! e^{
j2\pi\left(\Delta f\left(-\tau_{1,1}+\tau_{2,1}-\tau_{2,3}+\tau_{1,3}\right)+n_{2,1}-n_{2,3}+n_{1,3}\right)
 } \! \right|
\\
\label{eq:effAmp_los1}
&\stackrel{\text{(c)}}{=}&  \left|h_{1,1}\right| \left| \sin\!
\left(\pi n \Delta f_\text{min} \Delta\tau_1 \right)\right|
\end{IEEEeqnarray}
where $\Delta\tau_1 = -\tau_{1,1}+\tau_{2,1}-\tau_{2,3}+\tau_{1,3}$. For (a) we used \eqref{eq:amp_larg}. For (b) we inserted \eqref{eq:x_v1}, into which we inserted \eqref{eq:x_u2} and \eqref{eq:x_v3}. For (c) we used $\Delta f =  n \Delta f_\text{min}$, $\left|1 - e^{j\theta }\right|=2\left|\sin\!\left(\theta/2\right)\right|$ and $\left|\sin\!\left(\theta+\pi l\right)\right|=\left|\sin\!\left(\theta\right)\right|$ for $l\in\mathbb{Z}$.

The amplitudes of the second and third direct channels follow similarly and are
\begin{IEEEeqnarray}{rcl}
\notag
\left|\bar{h}_{2}\right|
\label{eq:effAmp_los2}
&=& \left|h_{2,2}\right| \left| \sin\!
\left(\pi n \Delta f_\text{min} \Delta\tau_2 \right)\right|
\\
\left|\bar{h}_{3}\right|
\label{eq:effAmp_los3}
&=& \left|h_{3,3}\right| \left|\sin\!
\left(\pi n \Delta f_\text{min} \Delta\tau_3 \right)\right|
\end{IEEEeqnarray}
where $\Delta\tau_2 = -\tau_{2,2}+\tau_{2,3}-\tau_{1,3}+\tau_{1,2}$ and $\Delta\tau_3 = -\tau_{3,3}+\tau_{3,2}-\tau_{1,2}+\tau_{1,3}$.

Examining the effective channel amplitudes, we observe that the amplitude of the $i$-th channel is bounded by
\begin{IEEEeqnarray}{rCcCl}
0 &\leq& \left|\bar{h}_{i}\right| &\leq& \left|h_{i,i}\right|.
\end{IEEEeqnarray}
For a given channel one can influence only the integer $n$ of the argument of the sine function, as the $\Delta\tau$ and the $\Delta f_\text{min}$ are fixed.

\subsection{Upper Bound}
The sum-rate of the proposed scheme for a three user pairs system with line-of-sight channels is upper bounded by
\begin{IEEEeqnarray}{rCl}
R_\text{sum} &\leq& \sum_{\forall i} \log_2 \left(1 + \frac{\left|\bar{h}_{i}\right|^2}{\sigma^2_{i}}\right).
\label{eq:IA_los_sum_UpperBound}
\end{IEEEeqnarray}
Since the sum-rate is different for different choices of $n$, one can optimize the choice of $\Delta f = n \Delta f_\text{min}$ within the available bandwidth to obtain the optimal sum-rate.

\begin{lemma}
 \label{lem:conv_upB}
 For continuously and independently distributed delays the upper bound on the sum-rate of the presented scheme is achieved arbitrarily closely for increasing bandwidth.
\end{lemma}

\begin{proof}
%
%
 The minimal sub-carrier spacing depends only on the delays and the delays are continuously and independently distributed. Hence also the products $\lambda_i=\Delta f_\text{min}\Delta\tau_i$ are continuously distributed. They are even independently distributed, since $\tau_{i,i}$ appears only in $\Delta\tau_i$. We can write $\lambda_i \mod 1$ with its infinitely long decimal expansion as
   \begin{IEEEeqnarray}{rCl}
  \lambda_i \mod 1 &=& 0.\lambda_i^{\left[1\right]}\lambda_i^{\left[2\right]}\lambda_i^{\left[3\right]} \ldots,
     \end{IEEEeqnarray}
 where each element $\lambda_i^{\left[l\right]}$ of the sequence is i.i.d. and takes on the values $\left\{0,1,2,\ldots9\right\}$ with equal probability.

 We wish to show that $\exists n \in \left\{ \mathbb{Z}:0<n<N\right\}$ with $N\rightarrow \infty$ such that $\left( n \lambda_i \mod 1 \right)$ $\forall i$ is arbitrarily close to some number $\mu\in(0,1)$.
 We do this by looking for strings of decimal places of $\lambda_i$ which are equal for all $i$ and which are, when shifted to the first decimal places, close enough to the desired number $\mu$. We then choose $n$ to shift the resulting sequence to the first decimal places.
 
 We choose $R\in\mathbb{Z}$ such that $10^{-R}<\epsilon$, where $0<\epsilon<1$. Our goal is to find an $r$ such, that the random variables $\mathcal{M}_r = \left\{ \lambda_i^{\left[w\right]}:\forall i, w=r,r+1,\ldots,r+R-1 \right\}$ fulfill the condition
    \begin{IEEEeqnarray}{rCl}
 \label{eq:dec_exp}
  \lambda_i^{\left[r\right]} \lambda_i^{\left[r+1\right]} \ldots \lambda_i^{\left[r+R-1\right]} &=& \mu^{\left[1\right]} \mu^{\left[2\right]} \ldots \mu^{\left[R\right]}, \quad \forall i,
     \end{IEEEeqnarray}
 where $\mu^{\left[w\right]}$ is the $w$-th position of the decimal expansion of $\mu$. The probability that the variables $\mathcal{M}_r$ fulfill the conditions \eqref{eq:dec_exp} for a given $r$ is positive. There are infinite independent realizations of the set $\mathcal{M}_r$, hence $\exists r$ such that the set $\mathcal{M}_r$ fulfills conditions \eqref{eq:dec_exp}.
 We complete the proof by choosing $n = 10^r$ and $\mu=1/2$.
\end{proof}

Lemma \ref{lem:conv_upB} ensures that by increasing the bandwidth and optimizing the choice of $\Delta f = n \Delta f_\text{min}$ we can get arbitrarily close to the upper bound of the presented scheme.

\subsection{Connection to Time Based Interference Alignment}
 Time based IA aligns interference by transmitting only in every other time slot and by (possibly) using different offsets. Interference is aligned when at the receivers the interference arrives in the same time slot, while the useful signals arrive in a different time slot. Analyses of time based IA can be found in \cite{Cadambe2007}, \cite{Grokop2008} or \cite{Blasco2011} for example. 
  
 We show that time based IA is a special case of subcarrier IA.
Choosing a precoder $\vec{v}_k$ in the frequency domain translates to the time domain signal
\begin{IEEEeqnarray}{rcccl}
 \label{eq:timeIA}
 \begin{bmatrix}
    \mathcal{X}_{k}[t]\\
		\mathcal{X}_{k}[t+1]
 \end{bmatrix}
 &=&
 \underbrace{
 \begin{bmatrix}
    1 & 1\\
		1 & -1
 \end{bmatrix}}_{\mat{F}^\dagger}
 \begin{bmatrix}
    v_{k}^{(1)}\\
		v_{k}^{(2)}
 \end{bmatrix}
 s_k
 &=&
 \begin{bmatrix}
    \left(v_{k}^{(1)} +	v_{k}^{(2)}\right)s_k \\
    \left(v_{k}^{(1)} -	v_{k}^{(2)}\right)s_k
 \end{bmatrix}\quad
\end{IEEEeqnarray}
where $\mat{F}^\dagger$ is the IDFT matrix.
Since for time based IA nothing is transmitted in the second time slot, we have $v_{k}^{(1)} -	v_{k}^{(2)} = 0$. Thus $\ln\!\left( v^{(2)}_{k}/v^{(1)}_{k}\right)=0$, $\forall k$ follows. In a similar way we obtain $\ln\!\left( u^{(2)\ast}_{i}/u^{(1)\ast}_{i}\right)=0$, $\forall i$.
This means that the right-hand side in \eqref{eq:IA_linSys} must be $\vec{b}=\vec{0}$, which automatically fulfills \eqref{eq:con_3user} and hence \eqref{eq:th_3user} and \eqref{eq:IA3user_los}. From $\vec{b}=\vec{0}$ it follows that
\begin{IEEEeqnarray}{rCcCl}
j\pi\!\left(1+2n_{i,k}\right) &=& \ln\!\left(\! \frac{h_{i,k}^{(2)}}{h_{i,k}^{(1)}}\!\right) &=& -j2\pi\Delta f \tau_{i,k}, \quad \forall i \neq k\quad
\end{IEEEeqnarray}
from where we obtain the conditions on the subcarrier spacing
\begin{IEEEeqnarray}{rCl}
\Delta f = \frac{1+2n_{i,k}}{2 \tau_{i,k}}, \quad \forall i \neq k.\quad
\end{IEEEeqnarray}
For $K=3$ there are six fractions that must be equal to each other and which determine $\Delta f$. The denominators of the fractions are real numbers while the numerators are integers. Since the delays are i.i.d., equality of these fractions is approached only by choosing larger integer numerators. This means that feasibility is achieved only asymptotically for increasing $\Delta f$, which translates to decreasing slot lengths in the time domain. This is precisely what Theorem 1 in \cite{Cadambe2007} states.
But we are able to determine subcarrier spacings which achieve feasibility \emph{exactly} for $K=3$. This shows that restricting the choice of the precoder, as time based IA does, prohibits achieving the DoF exactly.

\section{Simulation Results}
\label{sec:simu_res}
Consider a 3 user pair line-of-sight channel, where the transmitter-receiver distances $d_{i,k}$ are continuously and independently distributed. The delays are related to the distances by
  \begin{IEEEeqnarray}{rCl}
\tau_{i,k} = \frac{c}{d_{i,k}}
  \end{IEEEeqnarray}
where $c$ is the speed of wave propagation, which we set to the speed of light $c=3 \cdot 10^8 \,\text{m}/\text{s}$. The channel amplitudes are obtained from the distances as
  \begin{IEEEeqnarray}{rCl}
\left|h_{i,k}\right| = \left(\frac{1\text{m}}{d_{i,k}}\right)^\gamma
  \end{IEEEeqnarray}
where we choose the path-loss exponent $\gamma = 3.76$.

The distances of the direct channels are distributed as $d_{i,i}\in[150\text{m},250\text{m}]$, and the distances of the cross channels as $d_{i,k}\in[250\text{m},350\text{m}]$, $i\neq k$. The direct channels thus have the largest amplitudes and we do not have too small distances (for which treating interference as noise works best). We average over $10^4$ channel realizations.

As benchmark schemes we consider (I) treating \emph{Interference as Noise} and (II) an orthogonal access scheme, where we use \emph{TDMA}. For treating \emph{Interference as Noise}, each transmitter transmits two streams for every channel use and at the receivers the interference is treated as noise. For the \emph{TDMA} scheme, each transmitter transmits only in every $K$-th slot, but with $K$ times the power. Since only one pair communicates per slot, the receiver can receive two streams without interference.

To obtain the precoder and receive filter for IA, we use the pseudo-inverse of $\mat{A}$ to obtain a solution (or a least-squares solution, if IA is infeasible) for the system of linear equations \eqref{eq:IA_linSys}.
Since we are interested mainly in the DoF, we consider only interference-zero-forcing approaches. Other approaches, e.g. Max\-SINR or MMSE, will be examined in future work.

The values of $\Delta f_\text{min}$ seem to be Rayleigh-distributed, where more than $95\%$ of the values are between $10^6 \text{Hz}$ and $10^8 \text{Hz}$ for the considered scenario. These values depend strongly on the distances and the speed of wave propagation. For increasing distances or decreasing $c$ (e.g. under-water communication) the distribution of $\Delta f_\text{min}$ is shifted to lower frequencies.

Figure \ref{fig:Rsum_conv} shows the average sum-rate of the benchmark schemes and of IA for an average received SNR from the direct channels of $20\text{dB}$. The x-axis is normalized to $1/\Delta f_\text{min}$, where $\Delta f_\text{min}$ is different for every channel realization. As expected, the benchmark schemes perform independent of the subcarrier spacing. For IA we plot three curves.
The curve labeled \emph{IA ZF} is the average sum-rate with the current subcarrier spacing. As expected, we observe peaks at multiples of $\Delta f_\text{min}$. Note that for small deviations from the optimal $\Delta f_\text{min}$ there are small reductions in sum-rate. A subcarrier spacing between multiples of $\Delta f_\text{min}$ leads to leakage interference, which prevents achieving the maximal DoF. But for finite SNR we achieve a good performance when the direct channel's amplitude is large. The curve labeled \emph{Max IA ZF} is obtained in two steps: For each channel realization the maximal sum-rate within the bandwidth equal to the x-axis' value is determined. In the next step we take the average and obtain the curve labeled \emph{Max IA ZF}. A steep increase of this curve can be observed around $\Delta f_\text{min}$ due to the feasibility of IA. With increasing bandwidth the curve labeled \emph{Max IA ZF} approaches the curve labeled \emph{IA Upper Bound}, which is the average of the upper bounds \eqref{eq:IA_los_sum_UpperBound}.

\begin{figure}[htbp]
  \centering
  \psfrag{optimal delta}[][][.7]{$f/\Delta f_\text{min}$}
  \includegraphics[width=.48\textwidth]{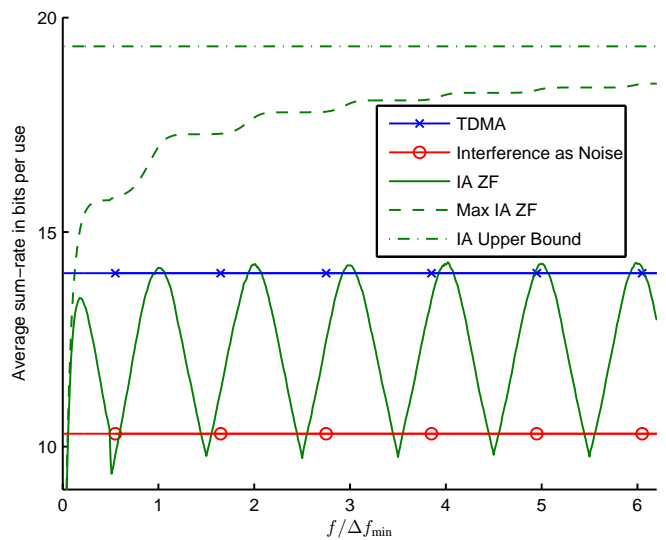}
  \caption{Average sum-rate for randomly distributed distances, where $d_{i,i}\in[150m,250m]$ and $d_{i,k}\in[250m,350m]$, $i\neq k$ and $\gamma = 3.76$ and the average received SNR from the direct channels is $20\text{dB}$.}
  \label{fig:Rsum_conv}
\end{figure}

\section{Conclusions}
\label{sec:conc}
We derived conditions for feasibility of IA via two orthogonal subcarriers. For line-of-sight channels these conditions can be fulfilled by carefully choosing the subcarrier spacing.


\section*{Acknowledgment}
S. Dierks and G. Kramer were supported by the German Ministry of Education and Research in the framework of an Alexander von Humboldt Professorship.
%

%
\bibliographystyle{IEEEtran}
\bibliography{/home/dierks/Literature/library}

\begin{thebibliography}{1}
\providecommand{\url}[1]{#1}
\csname url@samestyle\endcsname
\providecommand{\newblock}{\relax}
\providecommand{\bibinfo}[2]{#2}
\providecommand{\BIBentrySTDinterwordspacing}{\spaceskip=0pt\relax}
\providecommand{\BIBentryALTinterwordstretchfactor}{4}
\providecommand{\BIBentryALTinterwordspacing}{\spaceskip=\fontdimen2\font plus
\BIBentryALTinterwordstretchfactor\fontdimen3\font minus
  \fontdimen4\font\relax}
\providecommand{\BIBforeignlanguage}[2]{{%
\expandafter\ifx\csname l@#1\endcsname\relax
\typeout{** WARNING: IEEEtran.bst: No hyphenation pattern has been}%
\typeout{** loaded for the language `#1'. Using the pattern for}%
\typeout{** the default language instead.}%
\else
\language=\csname l@#1\endcsname
\fi
#2}}
\providecommand{\BIBdecl}{\relax}
\BIBdecl

\bibitem{Jafar2010}
\BIBentryALTinterwordspacing
S.~A. Jafar, \emph{{Interference Alignment — A New Look at Signal Dimensions
  in a Communication Network}}.\hskip 1em plus 0.5em minus 0.4em\relax
  Foundations and Trends® in Communications and Information Theory, 2010,
  vol.~7, no.~1. [Online]. Available: \url{www.nowpublishers.com}
\BIBentrySTDinterwordspacing

\bibitem{Cadambe2008}
V.~R. Cadambe and S.~A. Jafar, ``{Interference Alignment and Degrees of Freedom
  of the K-User Interference Channel},'' \emph{IEEE Trans. Inf. Theory},
  vol.~54, no.~8, pp. 3425--3441, Aug. 2008.

\bibitem{Gomadam2011}
K.~Gomadam, V.~R. Cadambe, and S.~A. Jafar, ``{A Distributed Numerical Approach
  to Interference Alignment and Applications to Wireless Interference
  Networks},'' \emph{IEEE Trans. Inf. Theory}, vol.~57, no.~6, pp. 3309--3322,
  Jun. 2011.

\bibitem{Cadambe2007}
V.~R. Cadambe and S.~A. Jafar, ``{Degrees of Freedom of Wireless Networks -
  What a Difference Delay Makes},'' in \emph{ACSSC 2007 - 41th Asilomar Conf.
  Signals, Sys. Comp.}, Nov. 2007, pp. 133--137.

\bibitem{Yetis2009}
C.~M. Yetis, T.~Gou, S.~A. Jafar, and A.~H. Kayran, ``{Feasibility Conditions
  for Interference Alignment},'' in \emph{GLOBECOM 2009 - IEEE Global
  Telecommun. Conf.}, Nov. 2009, pp. 1--6.

\bibitem{Grokop2008}
L.~H. Grokop, D.~N.~C. Tse, and R.~D. Yates, ``{Interference Alignment for
  Line-of-Sight Channels},'' \emph{IEEE Trans. Inf. Theory}, vol.~57, no.~9,
  pp. 5820--5839, Sep. 2011.

\bibitem{Blasco2011}
F.~{Lazaro Blasco}, F.~Rossetto, and G.~Bauch, ``{Time Interference Alignment
  via Delay Offset for Long Delay Networks},'' in \emph{GLOBECOM 2011 - IEEE
  Global Telecommun. Conf.}, Dec. 2011, pp. 1--6.

\end{thebibliography}
%
%

\end{document}